\documentclass[10pt,a4paper,reqno]{amsart}
\usepackage[latin1]{inputenc}
\usepackage{amsmath,enumerate,amsthm,graphicx,color,verbatim}
\usepackage{amsfonts,amscd,amsaddr}
\usepackage{latexsym}
\usepackage{amssymb}
\usepackage{breqn}
\newtheorem*{thm}{Theorem}

\newtheorem{defn}{Definition}

\newtheorem{prop}{Proposition}

\definecolor{orange}{rgb}{1,0.5,0}

\begin{document}
\title[description of Toda hierarchy using cocycles]{On a description of the Toda hierarchy using cocycle maps}
\author{Darren C. Ong}
\address{Department of Mathematics,\\ Xiamen University Malaysia,\\ Jalan
Sunsuria, Bandar Sunsuria,\\ 43900 Sepang, Selangor Darul Ehsan,\\ Malaysia}
\email{darrenong@xmu.edu.my}
\urladdr{https://dongcl.wixsite.com/darrenong}

\maketitle
\begin{abstract}The Toda hierarchy refers to a family of integrable flows on Jacobi operators that have many applications in mathematics and physics.  We demonstrate carefully that an alternative characterization of the Toda hierarchy using cocycle maps is equivalent to the traditional approach using Lax pairs.
\end{abstract}
\begin{section}{Introduction}
A Jacobi operator $J$ is a self-adjoint operator from $\ell^2(\mathbb Z)$ to $\ell^2(\mathbb Z)$. It is typically written as the tri-diagonal matrix,

\[
\begin{pmatrix}
\ddots&&&&\\
&b_1&a_1&&&\\
&a_1&b_2&a_2&&\\
&&a_2&b_3&a_3&\\
&&&a_3&b_4&\\
&&&&&\ddots
\end{pmatrix},
\]
where the $a_j$ are positive, the $b_j$ are real, and the index $j$ runs through $\mathbb Z$. We typically also insist that the $a_j$ and $b_j$ are uniformly bounded.

Alternatively, with $\{u_n\}\in\mathbb \ell^2(\mathbb Z)$ we may think of the Jacobi operator as  the following symmetric difference expression:

\[
J:u_n\to a_n u_{n+1}+a_{n-1}u_{n-1}+b_n u_n.
\]

Let $\mathcal J$ be the set of all Jacobi operators with operator norm less than or equal to $2$. We want to introduce a time-dependence on the Jacobi operator. To this end, we set an initial condition $J\in\mathcal J$ and define a differential $\mathbb R$-group action $\odot$ such that

\[t\odot J\in\mathcal J.\]
 
In particular, we will define $t\odot J$ to be a \emph{Toda flow}. This is a particular type of evolution on Jacobi operators that has many applications straddling both sides of the boundary between physics and mathematics. It was initially developed to describe the motion of particles in a nonlinear one-dimensional crystal with nearest neighbor interaction. From the perspective of inverse spectral theory, the Toda flow is important because it is an evolution on the Jacobi operators that preserves the spectrum. The Toda flow can also be regarded as a Jacobi operator analogue of the Korteweg-de Vries flow which relates to the evolution of waves in shallow media. 

The traditional mathematical treatment of the Toda flow defines it via the Lax equation:

\begin{equation}\label{LaxEquation}
\dfrac{d}{dt} (t\odot J)=P(t)(t\odot J)-(t\odot J)P(t),
\end{equation}
where $P(t)$ is a finite operator on $\ell^2(\mathbb Z)$ that satisfies $(P(t))^{tr}=-P(t)$. There are several appropriate choices for $P(t)$, each generating a different evolution of the Jacobi operators. In fact, it makes sense to parametrize the possible solutions $-P(t)$ by real polynomials with constant coefficient $1$,
\[\{1+c_1z+c_2z^2+\ldots c_r z^r\vert r\in\mathbb Z_+, c_j\in \mathbb C\}.\]
The \emph{Toda hierarchy} refers to all possible flows $t\odot J$ generated in this way. 

The purpose of this paper is to focus on an alternative characterization of the Toda hierarchy. We demonstrate carefully that this new definition is equivalent to the traditional one. This paper may be viewed as a comapanion result to \cite{Remling17} which discusses at greater length this alternative perspective for the Toda flow and a generalized flow on canonical systems. 

To motivate the alternative characterization, let us first discuss the role of the Weyl-Titchmarsh $m$-functions of the Jacobi operator. Each Jacobi operator $J$ corresponds to a pair of Herglotz functions (analytic functions from the upper half plane to itself), which we label as $m_+(z)$ and $m_-(z)$. There are a few different variants of these $m$-functions,  but we will use the definition

\[
m_\pm(z)=\mp \frac{u_1^\pm(z)}{a_0 u_0^\pm(z)},
\]
where $u_+$ is square-summable at $+\infty$, $u_-$ is square-summable at $-\infty$ and $u_\pm(z)$ are solutions to the difference equation $Ju=zu$ that are not necessarily in $\ell^2(\mathbb Z)$ (we refer to $u$ in this case as a \emph{generalized eigenfunction}).
	
These Weyl-Titchmarsh $m$-functions are of great importance in spectral theory. Their limiting behavior on $\mathbb R$ gives us information on the spectral measures of the Jacobi operator. For example, the pure points of the spectral measure are associated with points $x$ on $\mathbb R$ for which the limit of $m_\pm(x+i\epsilon)$ tends to a pole as $\epsilon$ tends to $0$.

In particular, it is often useful to treat the Toda flow as an evolution on the Weyl-Titchmarsh $m$-functions, rather than as an evolution on the Jacobi operator itself. See  \cite{Remling-Generalized}, \cite{Remling17} and \cite{Ong-Remling} for examples of this approach. Given a Toda flow, we thus write $m_\pm(z,t)$ as the $m$-functions corresponding to $t\odot J$.

Let us consider an $\mathrm{SL}(2,\mathbb C)$ matrix action on a complex number (treating the complex number as an element of complex projective space) as a linear fractional transformation,
that is

\[
\begin{pmatrix}
a&b\\
c& d
\end{pmatrix}z:=\frac{az+b}{cz+d}.
\]

There exists an $\mathrm{SL}(2,\mathbb C)$ matrix $T(z, J)$ such that

\begin{equation}\label{Tm_-}
m_-(z,t)=T(z,t\odot J)m_-(z,0),
\end{equation}
and 
\begin{equation}\label{Tm_+} m_+(z,t)=\begin{pmatrix}
1&0\\
0& -1
\end{pmatrix}T(z,t\odot J)
\begin{pmatrix}
1&0\\
0& -1
\end{pmatrix}
m_+(z,0).
\end{equation}

This just follows from the transfer matrix formalism for solutions of the Jacobi difference equation, and for a thorough treatment we refer the reader to \cite{Gesztesy-Holden-discrete} and \cite{Teschl-Jacobi}.

By theorem 2.2 of \cite{Remling17}, we can find a $T$ that is a cocyle, that is, it obeys for any $s,t,\in\mathbb R $,

\begin{equation}\label{cocycle}
T(s+t,J)=T(s,t\odot J)T(t,J).
\end{equation}

Recall that the one-step transfer matrix for a Jacobi operator takes the form

\begin{equation}\label{shiftmatrix}
M(J)=\begin{pmatrix}
\frac{z-b_1}{a_1}& \frac{1}{a_1}\\
-a_1 & 0
\end{pmatrix}
\end{equation}

Let $\{u_n\}$ be any generalized eigenvector of the Jacobi recursion. Then the transfer matrix $M$ satisfies

\[ M(J)\begin{pmatrix}
-u_1\\
a_0u_0
\end{pmatrix}=
\begin{pmatrix}
-u_2\\
a_1u_1
\end{pmatrix}.
\]

It is well known that the Toda flow and the left shift commute. In terms of $T$, this fact is expressed as the relation
\begin{equation}\label{shiftflowcommute}
M(t\odot J)T(t,J)=T(t,SJ)M(J),
\end{equation}
where $SJ$ refers to our Jacobi operator $J$ shifted once to the left. In $SJ$ we take $J$ and have every term $a_n, b_n$ replaced with $a_{n+1},b_{n+1}$ for all $n\in\mathbb Z$. More precisely, if $L$ is the operator that takes the vector $\{u_n\}_{n\in\mathbb Z}$ to $\{u_{n+1}\}_{n\in\mathbb Z}$, $SJ$ is shorthand for the operator $L^*JL$.

Given a $\mathrm{SL}(2,\mathbb C)$ cocycle $T$ that satisfies \eqref{cocycle}, it is known that there exists a trace zero matrix $B(J)$ such that

\[
\dfrac{d}{dt} T(t,J)=B(t\odot J)T(t,J).
\]

It turns out that for the Toda flow, the resulting $B$ matrix will always have entries that are continuous with respect to $J$ given the operator norm on $\mathcal J$, and polynomial in $z$. Our main theorem is that these properties, along with the commutativity of the Toda flow with the left shift on entries of the Jacobi operator, characterize the Toda hierarchy. 
\begin{thm}\label{maintheorem}
Let $t\odot J$ be a differentiable $\mathbb R$-group action on $\mathcal J$. Let $T(t,J)$ be an $\mathrm{SL}(2,\mathbb C)$ cocycle with respect to this group action $\odot$. Suppose also that $T(t,J)$ satisfies \eqref{shiftflowcommute}. This $T(t,J)$ in turn corresponds to a trace zero $2\times 2$ matrix $B(J)$. 

Suppose that the entries of the $B(J)$ matrix are polynomials in $z$ and continuous with respect to $J$. In fact, we take the perspective that our choice of such $B(J)$ determines $T(t,J)$ and hence the group action $t\odot J$. Then $t\odot J$ is the Toda flow with initial condition $J$ corresponding to the polynomial $p_dz^d+\ldots+p_1z+1$, with the $p_j$s determined by $B(J)$. In particular, by varying our choice of $B(J)$ we can generate every member of the Toda hierarchy in this way.
\end{thm}
\begin{paragraph}{\textbf{Remarks.}} \begin{enumerate}[(i)]
\item We only assume in this theorem that $T$ is an $\mathrm{SL}(2,\mathbb C)$ cocycle (i.e., it obeys \eqref{cocycle}) that commutes with the left shift. We do not assume that $T$ obeys \eqref{Tm_-} and \eqref{Tm_+}. Rather, the fact that $T$ evolves the $m$-functions emerges an a consequence of our hypotheses on $B(J)$. More precisely, our hypotheses imply that we get a member of the Toda hierarchy, and we know that those flows evolve the $m$-functions correctly. This is a rather surprising result, as even though the conditions we place on $B(J)$ are strong, they do not seem immediately related to the fact that $T$ takes $m$-functions to $m$-functions.
\item As we alluded to earlier, it makes sense to write the possible solutions to \eqref{LaxEquation} in correspondence with real polynomials
 
\[
1+p_1z+p_2z^2+\ldots p_nz^n.
\] 
In our theorem, we can determine the $p_j$ from $B(J)$, using a recursion relation whose initial condition involves the top right entry of $B(J)$.
\item A somewhat analogous result for a different integrable hierarchy appears in Section 2.2 of \cite{GVY}.
\end{enumerate}
\end{paragraph}

We have thus introduced a way to define the Toda hierarchy in terms of continuity conditions on $B(J)$, rather than in terms of solutions of the Lax equation \eqref{LaxEquation}. This definition in terms of $B(J)$ is more natural if we are focused on the effect of the Toda flow on the Weyl-Titchmarsh $m$-functions.

We note that it is well-known that the Toda flow satisfies \eqref{shiftflowcommute}. This fact appears in for instance as (0.46) and (0.47) in \cite{Gesztesy-Holden-discrete}, and as (12.86) in \cite{Teschl-Jacobi} (later our Proposition \ref{prop:zerocurvature} will make this connection more clear). However, these traditional treatments of the Toda flow mention \eqref{shiftflowcommute} as a consequence of the Lax equation \eqref{LaxEquation}, whereas we start with \eqref{shiftflowcommute} and derive the Toda flow from it. 

In particular, \cite{Remling17} and \cite{Ong-Remling} demonstrate the usefulness of this perspective The paper \cite{Remling17} works out a Toda-type flow for \emph{canonical systems} (\cite{deBranges}). Canonical systems are a spectral problem that generalizes the Jacobi equation. Not every pair of Herglotz functions correspond to the $m$-functions of a Jacobi operator, but every pair of Herglotz functions correspond to the $m$-functions of a canonical system. In essence, canonical systems are a spectral problem centered on the $m$-functions, rather than on operators.  A Toda-type flow thus is more sensibly expressed as an action on the $m$-functions rather than in terms of a Lax-type equation on a matrix operator. As another example, \cite{Ong-Remling} uses this cocycle map point of view to develop a generalization of the Toda flow where each flow corresponds to a $C^2$ or a $C^\infty$ function rather than just a polynomial.

\begin{paragraph}{\textbf{Acknowledgements}} I wish to thank Christian Remling for helpful conversations. I was supported by a Xiamen University Malaysia Research Fund (Grant No: XMUMRF/2018-C1/IMAT/0001)
\end{paragraph}
\end{section}
\begin{section}{Preliminaries}
\begin{subsection}{The Toda hierarchy in terms of solutions of the Lax equation}

This subsection is a brief summary of the relevant background material about the Toda hiearchy. Please consult Chapter 12 of \cite{Teschl-Jacobi}, for a more extensive treatment.

We begin with Jacobi parameters that depend on time:

\[
a(t),b(t)\in\ell^\infty (\mathbb Z,\mathbb R), a(n,t)\neq 0, (n,t)\in \mathbb Z\times\mathbb R.
\]
Furthermore, we assume that the map 

\[t\in\mathbb R\to (a(t),b(t))\in\ell^\infty (\mathbb Z)\oplus \ell^\infty(\mathbb Z)
\]
is differentiable.

The Jacobi operator is defined as a map $J(t):\ell^2(\mathbb Z)\to\ell^2(\mathbb Z)$ that takes $f_n$ to
\begin{equation}\label{JacobiDiffEqDef}
a_n(t)f_{n+1}+a_{n-1}(t)f_{n-1}+b_n(t)f_n.
\end{equation}

Recall the Lax equation \eqref{LaxEquation}, and recall that $P$ is a finite operator on $\ell^2(\mathbb Z)$ such that the transpose of $P$ is $-P$. 

According to Theorem 12.2 of \cite{Teschl-Jacobi}, $P(t)$ must be expressible in the form

\[P(t)=\sum_{j=0}^r \left(
c_{r-j} \widetilde{P_j}(t)+d_{r-j}J(t)^{j+1}\right)+d_{r+1}\mathrm{id},\]
where $\widetilde {P_j}(t)$ is obtained by subtracting the lower triangular part of $J(t)^{j+1}$ from the upper triangular part of $J(t)^{j+1}$. For example,

\[
\widetilde{P_0}(t)=
\begin{pmatrix}
\ddots&&&&\\
&0&a_1&&&\\
&-a_1&0&a_2&&\\
&&-a_2&0&a_3&\\
&&&-a_3&0&\\
&&&&&\ddots
\end{pmatrix}.
\]

For simplicity's sake, let us assume all the $d_j=0$, since those terms just get immediately cancelled out in the right hand side of \eqref{LaxEquation}. Thus, in a manner of speaking the $P(t)$ (and hence the possible Toda flows $J(t)$) are parametrized by the set of real polynomials with constant coefficient $1$,

\[\{1+c_1z+c_2z^2+\ldots c_r z^r\vert r\in\mathbb Z_+, c_j\in \mathbb C\}.\]
This set of flows is known as the Toda hierarchy. Each member of the Toda hierarchy corresponds to a $P(t)$ that solves the Lax equation \eqref{LaxEquation} 

\end{subsection}

\begin{subsection}{The Toda flow as an $\mathbb R$-group action on the $m$-functions}
Let $\mathcal J$ be the set of Jacobi operators whose norm is less than or equal to $2$. Consider a differentiable $\mathbb R$-group action $\odot$ on $\mathcal J$. Let $T$ be a map from $(\mathbb R,\mathcal J)$ to $\mathbb{SL}(2,\mathbb C)$.

\begin{defn}We say that $T$ is a \emph{cocycle} of this group action $\odot$ if it satisfies \eqref{cocycle}.
\end{defn}
\begin{prop}\label{T'=BT}
For $T$ a cocycle of $\odot$ and $J\in\mathcal J$ we have
\[ \dfrac{d}{dt} T(t,J)=B(t\odot J)T(t,J),\]
where $B$ is a map from $\mathcal J$ to the space of $2\times 2$ complex matrices. In particular, $B$ depends on $t\odot J$, and not individually on $t$ alone or on $J$ alone.
\end{prop}
\begin{proof}
We start by differentiating \eqref{cocycle} with respect to $t$:
\begin{align*}
\dfrac{d}{dt}T(t+s,J)=&\left(\dfrac{d}{dt}T(t,s\odot J)\right)T(s,J)\\
\left(\dfrac{d}{dt}T(t+s,J)\right)T^{-1}(t+s,J)=&\left(\dfrac{d}{dt}T(t,s\odot J)\right)T(s,J)T^{-1}(s,J)T^{-1}(t,s\odot J)\\
=&\left(\dfrac{d}{dt}T(t,s\odot J)\right)T^{-1}(t,s\odot J).\\
\end{align*}
Now setting $t=0$ for this last equality, we have

\[B(s\odot J):=\left(\dfrac{d}{dt}T(s,J)\right)T^{-1}(s,J)=\left(\dfrac{d}{dt}T(0,s\odot J)\right)T^{-1}(0,s\odot J).\]

In other words, the matrix $B(s\odot J)$ indeed depends on the action of $s$ on $J$, but not on $s$ alone or $J$ alone.
\end{proof}

Let us also assume \eqref{shiftflowcommute}, that $T$ commutes with the left shift. We can then prove

\begin{prop}\label{prop:zerocurvature}
\[\dfrac{d}{dt}M(J)=B(SJ)M(J)-M(J)B(J)\]
\end{prop}
\begin{proof}We differentiate both sides of \eqref{shiftflowcommute} with respect to $t$ to get
\[ \left(\dfrac{d}{dt}M(t\odot J)\right)T(t,J)+M(t\odot J)\dfrac{d}{dt}T(t,J)=\left(\dfrac{d}{dt}T(t,SJ)\right)M(J).\]
Plugging in $t=0$, this becomes
\[ \left(\dfrac{d}{dt}M(J)\right)+M(J)\dfrac{d}{dt}T(0,J)=\left(\dfrac{d}{dt}T(0,SJ)\right)M(J).\]
Note that by the definition of $B$, 

\begin{equation}\label{eq:B=T'}
B(J)=\dfrac{d}{dt}T(0,J) \text{ and }  B(SJ)=\dfrac{d}{dt}T(0,SJ).
\end{equation} This concludes our proof of the proposition.
\end{proof}
\begin{prop}$B$ has trace zero.
\end{prop}
\begin{proof}
 Firstly, since $T(t,J)$ is an $\mathrm{SL}(2,\mathbb C)$ matrix, $\dfrac{d}{dt}\det(T(t,J))=0$. Secondly, $T(0,J)$ is the identity matrix. These two facts imply that the trace of $\dfrac{d}{dt}T(t,J)$ is zero when $t=0$. The proposition then follows from \eqref{eq:B=T'}. 
\end{proof}

For a Jacobi operator $J$, let us write the entries of $B(J)$ as follows:

\begin{equation}\label{eq:ACD-A}
\begin{pmatrix}
A(z,J)& C(z,J)\\
D(z,J)& -A(z,J)
\end{pmatrix}.
\end{equation}

A calculation using \eqref{shiftmatrix} and Proposition \ref{prop:zerocurvature} gets us

\begin{align}
D(z,SJ)=&-a_1^2C(z,J),\label{MasterEq1}\\
A(z,SJ)+A(z,J)=&-a_1'/a_1+(z-b_1)C(z,J),\label{MasterEq2}\\
\frac{z-b_1}{a_1}\left[A(z,SJ)-A(z,J)\right]=& \frac{-b_1'a_1-(z-b_1)a_1'}{a_1^2}+a_1C(z,SJ)+\frac{D(z,J)}{a_1}.\label{MasterEq3}
\end{align}

\end{subsection}
\end{section}

\begin{section}{Proof of the main theorem}
 We take the equations \eqref{MasterEq1}, \eqref{MasterEq2}, \eqref{MasterEq3}. We assume that $A(z,J), C(z,J), D(z,J)$ are all polynomials of degree $d$, that is
 
 \[A(z,J)=\sum_{j=0}^d A_j(J)z^j, C(z,J)=\sum_{j=0}^d C_j(J)z^j, D(z,J)=\sum_{j=0}^d D_j(J)z^j.\]

 From \eqref{MasterEq1} we have
 
 \begin{equation}\label{D-aC}
 D_j(SJ)=-a_1^2C_j(J),
 \end{equation}
 for $j=0,\ldots, d$.
 
 From comparing the $z^{d+1}$ coefficients of \eqref{MasterEq2} we have
 
 \begin{equation}\label{C=0}
 C_d(J)=0.
\end{equation} 
 Furthermore, \eqref{D-aC} implies that 
\begin{equation}\label{D=0}
D_d(J)=0
\end{equation}
as well.
 
Comparing the $z^d, z^{d-1},\ldots, z$ coefficients of \eqref{MasterEq2} gives us
\begin{equation}\label{A+A=C-bC}
A_j(SJ)+A_j(J)=C_{j-1}(J)-b_1C_j(J),
\end{equation}
for $j=1,\ldots d$. Then, comparing the constant coefficients of \eqref{MasterEq2} we get

\begin{equation}\label{A+A=C-bC,j=0}
A_0(SJ)+A_0(J)=-\frac{a_1'}{a_1}-b_1C_0(J).
\end{equation}
 
We now turn our attention toward \eqref{MasterEq3}. From the $z^{d+1}$-coefficients, we observe that
$A_d(SJ)-A_d(J)=0$. Together with \eqref{A+A=C-bC} and \eqref{C=0} we find that

\begin{equation}\label{A=A=C}
A_d(SJ)=A_d(J)=\frac{C_{d-1}(J)}{2}.
\end{equation}
This implies that $A_d$, and hence $C_{d-1}$, must be independent of shifts on $J$.

We now look at the $z^d, z^{d-1},\ldots, z^2$ coefficients of \eqref{MasterEq3} and observe that

\begin{equation}\label{eq3^j=2..d}
-\frac{b_1}{a_1}(A_j(SJ)-A_j(J))+\frac{1}{a_1}(A_{j-1}(SJ)-A_{j-1}(J))=a_1C_j(SJ)+\frac{D_j(J)}{a_1}.
\end{equation}
  
 Comparing the $z^1$-coefficients of \eqref{MasterEq3} we have
 
 \begin{equation}\label{eq3^j=1}
-\frac{b_1}{a_1}(A_1(SJ)-A_1(J))+\frac{1}{a_1}(A_{0}(SJ)-A_{0}(J))=a_1C_1(SJ)+\frac{D_1(J)}{a_1}-\frac{a_1'}{a_1^2}.
\end{equation}

Lastly, comparing the constant coefficients of \eqref{MasterEq3} we have

 \begin{equation}\label{eq3^j=0}
-\frac{b_1}{a_1}(A_0(SJ)-A_0(J))=a_1C_0(SJ)+\frac{D_0(J)}{a_1}-\frac{b_1'}{a_1}+\frac{b_1a_1'}{a_1^2}.
\end{equation}

Let us define, for a matrix operator $M$ and an integer $n$
\begin{equation}
(M)_n=\left<\delta_n,M\delta_n\right>.
\end{equation}

Recall that $L$ is the left shift operator that takes the vector $\{u_n\}_{n\in\mathbb Z}$ to $\{u_{n+1}\}_{n\in\mathbb Z}$, and let $L^*$ be its adjoint, the right shift operator. Note that $L$ is related to $SJ$ in that $SJ=L^*JL$.

We can then define auxiliary functions
\begin{equation}
G^{(r)}_1=\sum_{j=0}^{r-1}z^{r-1-j}(J^j)_1,
\end{equation}
and
\begin{equation}
H^{(r)}_1=z^r-(J^r)_1+2a_1\sum_{j=1}^{r-1}z^{r-1-j}(LJ^j)_1.
\end{equation}
Their purpose will be made clear later.

Let us define $J$-dependent variables $p_1, p_2,\ldots, p_d$, and $q_1, q_2,\ldots, q_d$ and also

\begin{equation}
G_1=\sum_{j=1}^d p_j(J) G^{(j)}_1, H_1=\sum_{j=1}^d q_j(J) H^{(j)}_1,
\end{equation}

Alternatively, we may write $G_1$ as follows:

\begin{align}
\nonumber G_1=&z^{d-1}p_d(J)\\
\nonumber  &+z^{d-2}(p_d(J)\cdot (J)_1+p_{d-1}(J))\\
\nonumber &+z^{d-3}(p_d(J)\cdot(J^2)_1+p_{d-1}(J)\cdot(J)_1+p_{d-2}(J))\\
\nonumber  &+\ldots\\
\nonumber  &+z(p_d(J)\cdot(J^{d-2})_1+p_{d-1}(J)\cdot(J^{d-3})_1+\ldots+p_2(J))\\
\label{G_npolynomial}  &+p_d(J)\cdot(J^{d-1})_1+p_{d-1}(J)\cdot(J^{d-2})_1+\ldots+p_1(J)
\end{align}

Similarly we can write $H_1$ as

\begin{align}
\nonumber H_1=&z^{d}q_d(J)\\
\nonumber&+z^{d-1}q_{d-1}(J)\\
\nonumber&+z^{d-2}(q_{d-2}(J)+2a_1(SJ)_1q_d(J))\\
\nonumber&+z^{d-3}(q_{d-3}(J)+2a_1(LJ)_1q_{d-1}(J)+2a_1(LJ^2)_1q_{d}(J))\\
\nonumber&+\ldots\\
\nonumber&+z(q_1(J)+2a_1(LJ)_1q_3(J)+2a_1(LJ^2)_1q_4(J)+\ldots +2a_1(LJ^{d-2})_1q_{d}(J))\\
\nonumber&+2a_1((LJ)_1q_2(J)+(LJ^2)_1q_3(J)+\ldots +(LJ^{d-1})_1q_{d}(J))\\
&-(J)_1q_1(J)-(J^2)_1q_2(J)-\ldots-(J^d)_1q_d(J).
\label{H_npolynomial}
\end{align}

Let us define the $p_d,\ldots p_1$ recursively as follows:

\begin{align}
\label{p_d}p_{d}(J)=&C_{d-1}(J)/2,\\
\nonumber p_{d-1}(J)=&C_{d-2}(J)/2-p_d(J)\cdot(J)_1,\\
\nonumber p_{d-2}(J)=&C_{d-3}(J)/2-p_d(J)\cdot(J^2)_1-p_{d-1}(J)\cdot(J)_1,\\
\ldots\nonumber\\
\nonumber p_1(J)=&C_0(J)/2-p_d(J)\cdot(J^{d-1})_1-p_{d-1}(J)\cdot(J^{d-2})_1-\ldots-p_2(J)\cdot(J)_1
\end{align}

We also define $q_d,\ldots q_1$ as

\begin{align}
\nonumber q_d(J)=&-A_d(J)+C_{d-1}(J),\\
\nonumber q_{d-1}(J)=&-A_{d-1}(J)+C_{d-2}(J)-b_1C_{d-1}(J),\\
\nonumber q_{d-2}(J)=&-2a_1(LJ)_1q_d(J)-A_{d-2}(J)+C_{d-3}(J)-b_1C_{d-2}(J),\\
\nonumber q_{d-3}(J)=&-2a_1(LJ)_1q_{d-1}(J)-2a_1(LJ^2)_1q_{d}(J)\\
\nonumber&-A_{d-3}(J)+C_{d-4}(J)-b_1C_{d-3}(J),\\
\nonumber \ldots\\
\nonumber q_{1}(J)=&-2a_1(LJ)_1q_{3}(J)-2a_1(LJ^2)_1q_{4}(J)-\ldots -2a_1(LJ^{d-2})_1q_d(J)\\
\nonumber&-A_{1}(J)+C_{0}(J)-b_1C_{1}(J).\\ \label{q_d}
\end{align}

It is fairly easy to check that our definitions for $p_1,\ldots p_d$ imply
\begin{equation}\label{C=2Gequation}
C(z,J)=2G_1.
\end{equation} 

The definitions for $q_1,\ldots q_d$ then imply
\begin{align}
\nonumber A(z,J)=&(z-b_1)(2G_1)-H_1\\
\nonumber&+A_0(J)+b_1C_0(J)\\
\nonumber&+2a_1((LJ)_1q_2(J)+(LJ^2)_1q_3(J)+\ldots +(LJ^{d-1})_1q_{d}(J))\\
&-J_1q_1(J)-J^2_1q_2(J)-\ldots-J^d_1q_d(J).\label{A=G-Hequation}
\end{align}
The expression
\begin{align*}&A_0(J)+b_1C_0(J)
+2a_1((LJ)_1q_2(J)+(LJ^2)_1q_3(J)+\ldots +(LJ^{d-1})_1q_{d}(J))\\
&-J_1q_1(J)-J^2_1q_2(J)-\ldots-J^d_1q_d(J)
\end{align*}
is actually equal to zero, but we don't need that fact for now. We will perform this calculation at the end of the section.

Recalling \eqref{JacobiDiffEqDef} lets us write the Jacobi equation in the form

\[J=aL+L^{*}a+b.\]

In this expression, think of $a$ as the operator that takes the vector $\{u_n\}$ to $\{a_nu_n\}$, and $b$ as the operator that takes the vector $\{u_n\}$ to $\{b_nu_n\}$.

Recall that powers of $J$ are symmetric operators. Let us therefore note the following recursive expression for $(J^t)_1$: 

\begin{align}
\nonumber (J^t)_1=& \left<\delta_1,J^t\delta_1\right>\\
\nonumber=& \left<\delta_1,J^{t-1}J\delta_1\right>\\
\nonumber=& \left<\delta_1, J^{t-1}a_1\delta_2\right>+\left<\delta_1, J^{t-1}a_0\delta_0\right>+\left<\delta_1,J^{t-1}b_1 \delta_1\right>\\
\nonumber=& a_1\left<\delta_1, J^{t-1}L^*\delta_1\right>+a_0\left<L^*\delta_0, J^{t-1}\delta_0\right>+b_1\left<\delta_1,J^{t-1} \delta_1\right>\\
=&a_1(LJ^{t-1})_1+a_0(LJ^{t-1})_0+b_1 (J^{t-1})_1\label{<J>}
\end{align}

Similarly, we can calculate two recursive expressions for $(LJ^t)_1$:

\begin{align}
\nonumber(LJ^t)_1=& \left<\delta_1,J^t\delta_{2}\right>\\
\nonumber=& \left<\delta_1,JJ^{t-1}\delta_2\right>\\
\nonumber=& \left<\delta_1,aLJ^{t-1}\delta_2\right>+\left<\delta_1,L^*aJ^{t-1}\delta_2\right>
+\left<\delta_1,bJ^{t-1}\delta_2\right>\\
\nonumber=& \left<a_0\delta_0,J^{t-1}\delta_2\right>+\left<a_1\delta_2,J^{t-1}\delta_2\right>
+b_1\left<\delta_1,J^{t-1}\delta_2\right>\\
=&a_0 (L^2J^{t-1})_0+a_1(J^{t-1})_2+b_1(LJ^{t-1})_1\label{<SJ>1},
\end{align}
\begin{align}
\nonumber(LJ^t)_1=&\left<\delta_1,J^t\delta_2\right>\\
\nonumber=&\left<\delta_1,J^{t-1}J\delta_2\right>\\
\nonumber=& \left<\delta_1, J^{t-1}a_2\delta_{3}\right>+\left<\delta_1, J^{t-1}a_1\delta_1\right>+\left<\delta_1,J^{t-1}b_2 \delta_2\right>\\
=& a_2 (L^2J^{t-1})_1+(J^{t-1})_1a_1+(LJ^{t-1})_1b_2.\label{<SJ>2}
\end{align}

\begin{prop}\label{mainprop}
The $p_j, q_j,$ are independent of shifts on $J$ for all $j=1,\ldots d$. Furthermore, $p_j=q_j$ for all $j=1,\ldots d$.
\end{prop}
\begin{proof}
It is easy to check that this is true for $j=d$ by \eqref{A=A=C}, \eqref{p_d}, \eqref{q_d} and the fact that $A_d$ is shift-independent. We will proceed by induction.

As an inductive hypothesis, assume that the proposition holds true for $j=d,d-1,\ldots, k+1$. We will now prove it for $j=k$.

First note that since we are trying to prove the proposition for $j=1,\ldots d$, we may assume that $k\geq 1$. Thus \eqref{A+A=C-bC} implies

\begin{equation}
A_{k}(SJ)+A_{k}(J)=C_{k-1}(J)-b_1C_k(J),\label{A+A=C-bC,j=k}
\end{equation}

As a notational convention for the rest of this proof, we will write $p_k(J)$,$p_k(SJ)$, $q_k(J)$, $q_k(SJ)$ to emphasize that $p_k$, $q_k$ are possibly dependent on shifts of $J$, and for $j>k$ we will write $p_j$ and $q_j$ (leaving the $J$-dependence implicit) to emphasize that we know for sure that the $p_j$ and $q_j$ in question are independent of shifts (by the inductive hypothesis). Writing $A_{k}(J)$ down in terms of $p_j$s and $q_j$s, and recalling that for $j>k$ we know $p_j=q_j$, we have by  \eqref{q_d}, \eqref{C=2Gequation}, \eqref{G_npolynomial}, \eqref{H_npolynomial},

\begin{align*}
A_k(J)=&+2((J^{d-k})_{1}p_d+(J^{d-k-1})_1p_{d-1}+\ldots+ (J)_1 p_{k+1}+ p_k(J))\\
&-2b_1((J^{d-k-1})_1p_d+(J^{d-k-2})_1p_{d-1}+\ldots+ (J)_1 p_{k+2}+ p_{k+1})\\
&-q_k(J)-2a_1((LJ)_1p_{k+2}+(LJ^2)_1p_{k+3}
+\ldots+(LJ^{d-k-1})_1p_{d}),
\end{align*}
and
\begin{align*}
C_k(J)=& 2((J^{d-k-1})_1p_d+(J^{d-k-2})_1p_{d-1}+\ldots+ (J)_1 p_{k+2}+ p_{k+1}),\\
C_{k-1}(J)=& 2((J^{d-k})_{1}p_d+(J^{d-k-1})_1p_{d-1}+\ldots+ (J)_1 p_{k+1}+ p_k(J)).
\end{align*}

Let us plug these into \eqref{A+A=C-bC,j=k} and consider the terms in that equation that are a multiple of $p_t$, for some $t$ in $[k+2,d]$. These terms are

\begin{align*}
&2(J^{t-k})_2p_t-2b_2 (J^{t-k-1})_2 p_t-2a_2(LJ^{t-k-1})_2p_{t}\\
&+2(J^{t-k})_1p_t-2b_1 (J^{t-k-1})_1 p_t-2a_1(LJ^{t-k-1})_1p_{t}\\
-&2(J^{t-k})_1p_t+2b_1(J^{t-k-1})_1 p_t
\end{align*}

We can simplify this expression by making the obvious cancellations:

\begin{align*}
&2(J^{t-k})_2p_t-2b_2 (J^{t-k-1})_2p_t-2a_2(LJ^{t-k-1})_2p_{t}\\
&-2a_1(LJ^{t-k-1})_1p_{t}.
\end{align*}
But this is zero, by \eqref{<J>}.

For the case $t=k+1$, we get instead
 
 \begin{align*}
&2(J^{t-k})_2p_t-2b_2 (J^{t-k-1})_2 p_t\\
&+2(J^{t-k})_1p_t-2b_1 (J^{t-k-1})_1 p_t\\
-&2(J^{t-k})_1p_t+2b_1(J^{t-k-1})_1 p_t,
\end{align*}
which reduces to
\begin{align*}
&2(J^{t-k})_2p_t-2b_2 (J^{t-k-1})_2p_t.
\end{align*}
This is equal to $2b_2p_t-2b_2p_t=0$. 

Thus in \eqref{A+A=C-bC,j=k} we may ignore all the terms that are multiples of $p_t$, for any $t$ in $[k+1,d]$. This leaves us with

\[2p_k(SJ)-q_k(SJ)+2p_k(J)-q_k(J)=2p_k(J),\]
or equivalently 

\begin{equation}2p_k(SJ)-q_k(SJ)-q_k(J)=0.
\label{pq.equation1}
\end{equation}

Now we observe that \eqref{eq3^j=2..d} and \eqref{D-aC} imply that

\begin{equation}\label{eq3^j=k}
-b_1(A_{k+1}(SJ)-A_{k+1}(J))+(A_k(SJ)-A_k(J))=a_1^2C_{k+1}(SJ)-a_0^2C_{k+1}(S^{*}J),
\end{equation}
where $S^*J$ is shorthand for the operator $LJL^*$.

We can write this equation in terms of $p_j$s and $q_j$s. Let us move all the terms in \eqref{eq3^j=k} to the left and consider again only  the $p_t$ terms, for some $t$ in $[k+3,d]$. We end up with

\begin{align}
\nonumber &-b_1(2(J^{t-k-1})_2p_t-2b_2 (J^{t-k-2})_2p_t-2a_2(LJ^{t-k-2})_2p_{t})\\
\nonumber &+b_1(2(J^{t-k-1})_1p_t-2b_1 (J^{t-k-2})_1p_t-2a_1(LJ^{t-k-2})_1p_{t})\\
\nonumber &+2(J^{t-k})_2p_t-2b_2 (J^{t-k-1})_2p_t-2a_2(LJ^{t-k-1})_2p_{t}\\
\nonumber &-2(J^{t-k})_1p_t+2b_1 (J^{t-k-1})_1p_t+2a_1(LJ^{t-k-1})_1p_{t}\\
\nonumber &-2a_1^2(J^{t-k-2})_2p_t+2a_0^2(J^{t-k-2})_0p_t.\\
\label{longequation}
\end{align}

This seems like a formidable expression, but again we can simplify it. First, we apply \eqref{<J>} to the first term of each of the first four lines to get

\begin{align*}
&-b_1(2a_1(LJ^{t-k-2})_1p_t\\
&+b_1(2a_0(LJ^{t-k-2})_0p_t\\
&+2a_1(LJ^{t-k-1})_1p_t\\
&-2a_0(LJ^{t-k-1})_0p_t\\
&-2a_1^2(J^{t-k-2})_2p_t+2a_0^2(J^{t-k-2})_0p_t.
\end{align*}
Finally, applying \eqref{<SJ>1} to the $(LJ^{t-k-1})_1$ term, and \eqref{<SJ>2} to the $(LJ^{t-k-1})_0$ term, we find that everything cancels. Thus all the $p_t$ terms in \eqref{eq3^j=k} cancel out for $t=k+3,\ldots d$.

Let us consider $t=k+2$. Then instead of \eqref{longequation} we have

\begin{align}
\nonumber &-b_1(2(J^{t-k-1})_2p_t-2b_2 (J^{t-k-2})_2p_t)\\
\nonumber &+b_1(2(J^{t-k-1})_1p_t-2b_1 (J^{t-k-2})_1p_t)\\
\nonumber &+2(J^{t-k})_2p_t-2b_2 (J^{t-k-1})_2p_t-2a_2(LJ^{t-k-1})_2p_{t}\\
\nonumber &-2(J^{t-k})_1p_t+2b_1 (J^{t-k-1})_1p_t+2a_1(LJ^{t-k-1})_1p_{t}\\
\nonumber &-2a_1^2(J^{t-k-2})_2p_t+2a_0^2(J^{t-k-2})_0p_t.
\end{align}

Replacing all the $t$'s in this expression with $k+2$, we get

\begin{align}
\nonumber &+2(J^{2})_2p_t-2b_2 (J)_2p_t-2a_2(LJ)_2p_{t}\\
\nonumber &-2(J^{2})_1p_t+2b_1 (J)_1p_t+2a_1(LJ)_1p_{t}\\
\nonumber &-2a_1^2 p_t+2a_0^2p_t,
\end{align}
and this reduces to $0$ once we apply \eqref{<J>} to the $(J^{2})_2$ and $(J^{2})_1$ terms.

Finally, we consider the case $t=k+1$. Now instead of \eqref{longequation} we have

\begin{align}
\nonumber &-2b_1(J^{t-k-1})_2p_t\\
\nonumber &+2b_1(J^{t-k-1})_1p_t\\
\nonumber &+2(J^{t-k})_2p_t-2b_2 (J^{t-k-1})_2p_t\\
\nonumber &-2(J^{t-k})_1p_t+2b_1 (J^{t-k-1})_1p_t,
\end{align}
Replacing all the $t$'s with $k+1$ we get
\begin{align}
\nonumber &-2b_1p_t\\
\nonumber &+2b_1p_t\\
\nonumber &+2b_2p_t-2b_2 p_t\\
\nonumber &-2b_1p_t+2b_1 p_t,
\end{align}
which is clearly also zero. 

Thus all the $p_t$ terms in \eqref{eq3^j=k} cancel out for $t=k+1,\ldots d$. This means that \eqref{eq3^j=k} reduces to
\[2p_k(SJ)-q_k(SJ)-2p_k(J)+q_k(J)=0.\]
Together with \eqref{pq.equation1} this implies that $p_k(J)=q_k(J)$ and that $p_k,q_k$ are both independent of shifts on $J$, thus concluding our induction proof.
\end{proof}

Finally, one last bit of unfinished business. We have to clean up the extraneous terms in \eqref{A=G-Hequation}.
\begin{prop}\label{prop:unfinished}
\begin{align*}&A_0(J)+b_1C_0(J)
+2a_1((LJ)_1q_2+(LJ^2)_1q_3+\ldots +(LJ^{d-1})_1q_{d})\\
&-(J)_1q_1-(J^2)_1q_2-\ldots-(J^d)_1q_d=0
\end{align*}
\end{prop}
\begin{proof}
First, we note that by \eqref{C=2Gequation},\eqref{G_npolynomial} and Proposition \ref{mainprop},
\[C_0(J)=2(J^{d-1})_1q_d+2(J^{d-2})_1q_{d-1}+\ldots+2q_1.\]
 So it suffices to prove

\begin{align}\nonumber A_0(J)=& -2a_1((LJ)_1q_2+(LJ^2)_1q_3+\ldots +(LJ^{d-1})_1q_{d})\\
&+q_1((J)_1-2b_1)+q_2((J^2)_1-2b_1(J)_1)+\ldots+q_d((J^d)_1-2b_1(J^{d-1})_1).\label{A0goal}
\end{align}

Subtracting $a_1$ times \eqref{eq3^j=1} from
\eqref{A+A=C-bC,j=0} and using \eqref{D-aC} we obtain

\begin{equation}\label{A0equation}
2A_0(J)=-b_1C_0(J)-a_1^2C_1(SJ)+a_0^2C_1(S^{*}J)-b_1(A_1(SJ)-A_1(J))
\end{equation}

We already know how to write all the terms in the RHS in terms of the $q_j$. Let us concentrate on the terms which are a multiple of $q_t$, for some $t\in[2,\ldots,d]$. We extract the $q_t$ terms from all the expressions on the RHS 

\begin{align*}
-b_1C_0(J)\sim& -2b_1(J^{t-1})_1q_t\\
-a_1^2C_1(SJ)\sim& -2a_1^2(J^{t-2})_2q_t\\
a_0^2C_1(S^{*}J)\sim & 2a_0^2(J^{t-2})_0q_t\\
-b_1A_1(SJ)\sim& -b_1
(2(J^{t-1})_2-2b_2(J^{t-2})_2-2a_2(LJ^{t-2})_2)q_t\\
b_1A_1(J)\sim& b_1(2(J^{t-1})_1-2b_1(J^{t-2})_{1}-2a_1(LJ^{t-2})_1)q_t\\
\end{align*}

So when we multiply \eqref{A0equation} by $\frac{1}{2}$, that equation asserts that the $q_t$-term of $A_0(J)$ is

\begin{align*}
A_0(J)\sim & -b_1(J^{t-1})_1q_t
 -a_1^2(J^{t-2})_2q_t+
 a_0^2(J^{t-2})_0q_t\\
 &-b_1((J^{t-1})_2-b_2(J^{t-2})_2-a_2(LJ^{t-2})_2)q_t\\
 &+ b_1((J^{t-1})_1-b_1(J^{t-2})_{1}-a_1(LJ^{t-2})_1)q_t\\
\end{align*}

We may apply \eqref{<J>} to the $b_1(J^{t-1})_2$ and $b_1(J^{t-1})_1$ terms to simplify this expression:

\begin{align*}
A_0(J)\sim & -b_1(J^{t-1})_1q_t
 -a_1^2(J^{t-2})_2q_t+
 a_0^2(J^{t-2})_0q_t\\
 &-b_1a_1(LJ^{t-2})_1 q_t\\
 &+ b_1a_0(LJ^{t-2})_0q_t\\
\end{align*}

We may rewrite this as

\begin{align*}
A_0(J)\sim & -b_1(J^{t-1})_1 q_t
 -a_1 q_t(a_1(J^{t-2})_2+b_1(LJ^{t-2})_1)\\
& +a_0q_t(a_0(J^{t-2})_0+b_1(LJ)^{t-2}_0)\\
= & -b_1(J^{t-1})_1 q_t
 -a_1 q_t(a_1(J^{t-2})_2+b_1(LJ^{t-2})_1+a_0(L^2J^{t-1})_0)\\
& +a_0q_t(a_0(J^{t-2})_0+b_1(LJ^{t-2})_0+a_1(L^2J^{t-2})_0)\\
\end{align*}
We may then apply \eqref{<SJ>1} and \eqref{<SJ>2} to get

\begin{align*}
A_0(J)\sim & -b_1J^{t-1}_1 q_t
 -a_1(LJ^{t-1})_1q_t+a_0(LJ^{t-1})_0q_t\\
 =&a_1(LJ^{t-1})_1q_t+b_1(J^{t-1})_1q_t+a_0(LJ^{t-1})_0q_t-2b_1(J^{t-1})_1 q_t-2a_1(LJ^{t-1})_1q_t\\
\end{align*}

Applying \eqref{<J>} once more we get

\begin{align*}
A_0(J)\sim &(J^t)_1q_t-2b_1(J^{t-1})_1q_t-2a_1(LJ^{t-1})_1q_t.
\end{align*}

This matches the $q_t$ part of the right hand side of the equation \eqref{A0goal}, so we have proven that the $q_t$ terms in \eqref{A0goal} are correct for $t=2,\ldots, d$. All that remains is to show that the $q_1$ term is correct. Observe that \eqref{A0goal} asserts that the $q_1$ term of $A_0(J)$ is just $-q_1b_1$. It is easy to check that this follows from \eqref{A0equation} since the $q_1$ term of $-b_1C_0(J)$ is $-2b_1q_1$, and the expressions $C_1(SJ), C_1(S^{*}(J)$ and $(A_1(SJ)-A_1(J))$ do not have $q_1$ terms.
\end{proof}

\begin{prop}\label{J-independence}
If the entries of the $B(J)$ matrix are polynomials in $z$ that are continuous with respect to $J$, then $p_1,\ldots, p_d$ are independent of $J$ (as opposed to only being independent of shifts on $J$)
\end{prop}
\begin{proof}
We begin by choosing a Jacobi operator $J'$ whose orbit on the shift map is dense. Using \eqref{p_d} and the hypothesis that the entries of the $B(J)$ matrix are continuous in $J$, we note that starting with $J'$ and then shifting, we get that the $p_1(J),\ldots p_d(J)$ must be independent of $J$.
\end{proof}

\begin{proof}[Proof of Theorem]

We have shown (from \eqref{eq:ACD-A},\eqref{D-aC}, \eqref{C=2Gequation}, \eqref{A=G-Hequation}, and Proposition \ref{prop:unfinished}) that our $B(J)$ matrix takes the form

\begin{equation}\label{B(J)final}
B(J)=
\begin{pmatrix}
2(z-b_1)G_1(J)-H_1(J) & 2G_1(J)\\
-2a_0^2G_1(S^{*}J)& -2(z-b_1)G_1(J)+H_1(J)
\end{pmatrix}.
\end{equation}

Now let us consider the $B$-matrix we get when we derive the Toda hierarchy in the traditional way, using the Lax equation.  We can relate such a $B$ with the matrix $C_r$ defined in \cite[(12.93)]{Teschl-Jacobi} which satisfies

\[\dfrac{d}{dt}
\begin{pmatrix}
 u_0(z,t)\\
u_1(z,t)
\end{pmatrix}
=
-C_r(z,t)
\begin{pmatrix}
 u_0(z,t)\\
 u_1(z,t)
\end{pmatrix}.
\]

Note that we have from \eqref{Tm_-}

\[
\begin{pmatrix}
 u_0(z,t)\\
 u_1(z,t)
\end{pmatrix}=
\begin{pmatrix}
0&\frac{1}{a_0(t)}\\
-1&0
\end{pmatrix}
T(t,J)
\begin{pmatrix}
 -u_1(z,0)\\
 a_0(0)u_0(z,0)
\end{pmatrix},
\]
and therefore from Proposition \ref{T'=BT} we can derive an explicit relationship between $B$ and $C_r$:

\begin{equation}\label{B-Cr}
B=\begin{pmatrix}
0&0\\
0&\frac{\dfrac{d}{dt}a_0(t)}{a_0(t)}
\end{pmatrix}
-\begin{pmatrix}
0&-1\\
a_0(t)&0
\end{pmatrix}
C_r(z,t)
\begin{pmatrix}
0&\frac{1}{a_0(t)}\\
-1&0
\end{pmatrix}.
\end{equation}

Solving this equation and using \cite[(12.58) and (12.83)]{Teschl-Jacobi}, we get precisely the $B(J)$-matrix in \eqref{B(J)final} corresponding to polynomial $p_dz^d+\ldots+p_1z+1$. (The reader should note that there is a sign error in \cite[(12.83)]{Teschl-Jacobi}). Thus we obtain the Toda flow corresponding to the polynomial $c_dz^d+\ldots+c_1z+1$ by choosing $p_j=c_j$.

\end{proof}

\end{section}
\bibliographystyle{alpha}
\bibliography{../mybib}

\begin{thebibliography}{GHMT08}

\bibitem[dB68]{deBranges}
L.~{d}e Branges.
\newblock {\em Hilbert spaces of entire functions}.
\newblock Prentice-Hall, 1968.

\bibitem[GHMT08]{Gesztesy-Holden-discrete}
F.~Gesztesy, H.~Holden, J.~Michor, and G.~Teschl.
\newblock {\em Soliton Equations and Their Algebro-Geometric Solutions,
  (1+1)-Dimensional Discrete Models}.
\newblock Cambridge Studies in Advanced Mathematics, 114, {Vol} {I}{I}.
  Cambridge University Press, 2008.

\bibitem[GVY08]{GVY}
Vladimir Gerdjikov, Gaetano Vilasi, and Alexandar~Borisov Yanovski.
\newblock {\em Integrable {H}amiltonian hierarchies: Spectral and geometric
  methods}, volume 748 of {\em Lecture notes in physics}.
\newblock Springer, 2008.

\bibitem[OR18]{Ong-Remling}
Darren~C. Ong and Christian Remling.
\newblock Generalized {T}oda flows.
\newblock arXiv:1801.03053, 2018.

\bibitem[Rem15]{Remling-Generalized}
Christian Remling.
\newblock Generalized reflection coefficients.
\newblock {\em Communications in {M}athematical {P}hysics}, 337:1011--1026,
  2015.

\bibitem[Rem17]{Remling17}
Christian Remling.
\newblock Toda maps, cocycles, and canonical systems.
\newblock arXiv:1712.00503, 2017.

\bibitem[Tes00]{Teschl-Jacobi}
G.~Teschl.
\newblock {\em Jacobi operators and completely integrable nonlinear lattices},
  volume~72 of {\em Mathematical Surveys and Monographs}.
\newblock American {M}athematical {S}ociety, 2000.

\end{thebibliography}
\end{document}